\documentclass[10.8pt, letter]{amsart}

\usepackage{amsmath}
\usepackage{amsfonts}
\usepackage{amssymb}
\usepackage[colorlinks=true, pdfstartview=FitV, linkcolor=blue, citecolor=blue, urlcolor=blue,pagebackref=false]{hyperref}
\usepackage[usenames]{color}

\numberwithin{equation}{section}
\newtheorem{theorem}{Theorem}[section]

\newtheorem{proposition}[theorem]{Proposition}
\newtheorem{lemma}[theorem]{Lemma}
\theoremstyle{definition}

\theoremstyle{remark}

\DeclareMathOperator{\diam}{diam}

\definecolor{orange}{rgb}{1,0.5,0}
\definecolor{green}{rgb}{0.4,0.8,0.5}

\title[ ]
      {Stability estimates for the inverse boundary value problem by partial Cauchy data}
\author [ ]
        {Ru-Yu Lai}
\thanks{Department of Mathematics, University of Washington, Seattle, WA 98195, USA. Email: rylai@uw.edu}

\begin{document}

\maketitle

\maketitle

\setcounter{tocdepth}{1}

\begin{abstract}
  We study the inverse conductivity problem with partial data in dimension $n\geq 3$. We derive stability estimates for this inverse problem if the conductivity has $C^{1,\sigma}(\overline\Omega)\cap H^{\frac{3}{2}+\sigma}(\Omega)$ regularity for $0<\sigma<1$.
\end{abstract}

\section{Introduction}
In 1980 A. P. Calder\text{\'{o}}n published a short paper entitled ``On an inverse boundary value problem'' \cite{C}.
This pioneer contribution motivated many developments in inverse problems, in particular in the construction of ``complex geometrical optics'' (CGO) solutions of partial differential equations
to solve inverse problems. The problem that Calder\text{\'{o}}n considered was whether one can determine the
electrical conductivity of a medium by making voltage and current measurements at the boundary of the medium.
This inverse method is known as \emph{Electrical Impedance Tomography} (EIT).
EIT arises not only in geophysical prospections (See \cite{ZK}), but also in medical imaging (See \cite{Holder}, \cite{HIMS}
and \cite{J}). We now describe more precisely the mathematical problem. Let $\Omega\subset\mathbb{R}^n$ be a
bounded domain with smooth boundary. The electrical conductivity of $\Omega$ is represented by a bounded and
positive function $\gamma(x)$. In the absence of sinks or sources of current, the equation for the potential
is given by
\begin{align*}
         \nabla\cdot \gamma\nabla u =0 \ \ \hbox{in $\Omega$}
\end{align*}
since, by Ohm's law, $\gamma\nabla u$ represents the current flux. Given $f\in H^{1/2}(\partial\Omega)$ on the boundary, the potential $u\in H^1(\Omega)$ solves the Dirichlet problem
\begin{align}\label{con}
\left\{
       \begin{array}{rl}
         \nabla\cdot \gamma\nabla u = 0\ \ &\hbox{in $\Omega$} \\
         u = f\ \ &\hbox{on $\partial\Omega$.}
       \end{array}
       \right.
\end{align}
The Dirichlet-to-Neumann map, or voltage-to-current map, is given by
$$
     \Lambda_\gamma f=\gamma\partial_\nu u|_{\partial\Omega},
$$
where $\partial_\nu u=\nu\cdot\nabla u$ and $\nu$ is the unit outer normal to $\partial\Omega$. The well-known inverse problem is to recover the conductivity $\gamma$ from the boundary measurement $\Lambda_\gamma$.

The uniqueness issue for $C^2$ conductivities was first settled by Sylvester and Uhlmann \cite{SU1}.
Later, the regularity of conductivity was relaxed to $3/2$ derivatives in some sense in \cite{BT} and \cite{PPU}.
Uniqueness for conductivities with conormal singularities in $C^{1,\varepsilon}$ was shown in \cite{GLU}. See \cite{U1} for the detailed development. Recently, Haberman and Tataru \cite{HT} extended the uniqueness result to $C^{1}$ conductivities or small in the $W^{1,\infty}$ norm. It is an open problem whether uniqueness holds in dimension $n\geq 3$ for Lipschitz or less regular conductivities.



For the stability result, in 1988, a log-type stability estimate was derived by Alessandrini \cite{A1}. Mandache \cite{M} has shown that this estimate is optimal. Later, Heck \cite{H} proved the stability for conductivities in $C^{1,\frac{1}{2}+\varepsilon}\cap H^{\frac{n}{2}+\varepsilon}$ with smooth boundary in 2009. For the case $\gamma\in C^{1, \varepsilon}, 0<\varepsilon<1$, Caro, Garc\'{i}a and Reyes used Haberman and Tataru's ideas to derive the stability result with Lipschitz boundary. For a review of stability issues in EIT see \cite{A2}.

All results mentioned above are concerned with the full data. In several applications in EIT one can only measure currents and voltages on part of the boundary. A general uniqueness result with partial data was first obtained by Bukhgeim and Uhlmann \cite{BU} when the Neumann data were taken on part of $\partial\Omega$ which is slightly larger than the half of the boundary. Their result was improved in \cite{KSU} where the Cauchy data can be taken on any part of the boundary. In \cite{BU} and \cite{KSU}, the conductivities are in $C^2$. The regularity assumption on the conductivity was relaxed to $C^{1,\frac{1}{2}+\varepsilon}, \varepsilon>0$ by Knudsen in \cite{K}. In 2012, Zhang \cite {Z} gave the uniqueness result with $C^1\cap H^{3/2}$ conductivities by using the idea in \cite{HT} and following the argument in \cite{K}.
The stability estimates for the uniqueness result of \cite{BU} were given by Heck and Wang in \cite{HW}. Heck and Wang proved the log-log type stability estimate with partial data. They improved their result to the log type stability in the paper \cite{HW1} in 2007 by considering special domains.

In this paper, we derive a log-log type stability estimate for less regular conductivities.
To state the main result, we first introduce several notations. Picking a $\eta\in S^{n-1}$ and letting $\varepsilon>0$, we define
$$
     \partial\Omega_{+,\varepsilon}=\{x\in\partial\Omega: \eta\cdot\nu(x)>\varepsilon\},\ \ \partial\Omega_{-,\varepsilon}=\partial\Omega\backslash\overline{\partial\Omega_{+,\varepsilon}}.
$$
The localized Dirichlet-to-Neumann map is given by
$$
     \tilde\Lambda_\gamma:f\mapsto \gamma\partial_\nu u|_{\partial\Omega_{-,\varepsilon}}.
$$
So $ \tilde\Lambda_\gamma$ is an operator from $H^{1/2}(\partial\Omega)$ to $\tilde H^{-1/2}( \partial\Omega_{-,\varepsilon})$, the restriction of $H^{-1/2}(\partial\Omega)$ onto $\partial\Omega_{-,\varepsilon}$. The operator norm of $ \tilde\Lambda_\gamma$ is denoted by $\| \tilde\Lambda_\gamma\|_{*}$.

\begin{theorem}
Let $\Omega\subset \mathbb{R}^n, n\geq 3$, be an open, bounded
domain with $C^2 $ boundary. Let $\gamma_j\in
C^{1,\sigma}(\overline\Omega)\cap H^{\frac{3}{2}+\sigma}(\Omega)$
with $0<\sigma<1$  such that $\gamma_j>\gamma_0>0$ and
$$
    \|\gamma_j\|_{C^{1,\sigma}(\overline\Omega)}+\|\gamma_j\|_{ H^{\frac{3}{2}+\sigma}(\Omega)}\leq M
$$
for $j=1,2$ and some constants $\gamma_0,\ M>0$. Suppose that
$$
\gamma_1=\gamma_2\ \ \ \hbox{and}\ \ \ \partial_\nu \gamma_1=\partial_\nu \gamma_2 \ \ \ \hbox{on}\ \overline{\partial\Omega_{+,\varepsilon}}.
$$
Then there exist constants $\theta, \tilde\theta, \tilde\sigma\in (0,1)$ and constant $K$ such that
\begin{align}\label{first}
\|\gamma_1-\gamma_2\|_{C^{0,\tilde\sigma}(\overline\Omega)}
\lesssim   \Big( \|\tilde \Lambda_{\gamma_1}-\tilde\Lambda_{\gamma_2}\|^{\theta}_{*}+ \|\tilde \Lambda_{\gamma_1}-\tilde\Lambda_{\gamma_2}\|^{1-\theta}_{*} +\frac{1}{K}\log{|\log \|\tilde\Lambda_{\gamma_1}-\tilde\Lambda_{\gamma_2}\|_{*}^\theta|}^{-\frac{1}{\tilde\theta}} \Big)^{\frac{\tilde\theta(1-\tilde\sigma)}{n}}.
\end{align}
\end{theorem}
Note that the symbol $\lesssim$ means that there exists a positive constant for which the estimate holds whenever the right hand side of the estimate is multiplied by that constant.

Along our discussion we follow a recent improvement of the classical
method introduced by Sylvester and Uhlmann in \cite{SU1} and based
on the construction of CGO solutions. This new improvement is due to
Haberman and Tataru (see \cite{HT}) and it has allowed us to improve
Heck and Wang's result in \cite{HW} relaxing the smoothness of the
coefficients and the smoothness of the boundary of the domain. To
deriving the estimate (\ref{first}), we adapt Zhang's argument
\cite{Z} to the case $\tilde\Lambda_{\gamma_1}\neq
\tilde\Lambda_{\gamma_2}$. Then we will get an estimate of the
Fourier transform of $q:=(ik)\nabla v+\nabla(\log
\sqrt{\gamma_1}+\log\sqrt{\gamma_2})\nabla v$ on some subset of
$\mathbb{R}^n$ where $v=\log \sqrt{\gamma_1}-\log\sqrt{\gamma_2}$.
Since $q$ can be treated as a compactly supported function, its
Fourier transform is real analytic. We use Vessella's stability
estimate for analytic continuation \cite{V} to our case here. This
idea was first introduced in \cite{HW} to get the log-log type
stability estimate with partial measurements.

\section{Preliminary result}\label{preliminary}
Let $n\geq 3$ and $\Omega\subset \mathbb{R}^n$ be an open bounded
domain with $C^2$ boundary $\partial\Omega$ throughout the paper.
Assume that $\gamma_j\in C^{1,\sigma}(\overline\Omega)\cap
H^{\frac{3}{2}+\sigma}(\Omega)$ with $0<\sigma<1$ and
$\gamma_j>\gamma_0>0$ for $j=1,2$. Let $\overline \Omega\subset B$.
We can extend $\gamma_j$ to be the function in $\mathbb{R}^n$ such
that $\gamma_j\in C^{1,\sigma}(\mathbb{R}^n)$ with positive lower
bound and $\gamma_{j}-1\in H^{\frac{3}{2}+\sigma}(\mathbb{R}^n)$
with $\text{supp}(\gamma_{j}-1)\subset \overline B$ .

Let $\Psi_t=t^n\Psi(tx)$ where $\Psi\in C^\infty_0(\mathbb{R}^n)$ supported on the unit ball and $\int \Psi=1$. Denote that $\phi=\log \gamma$ and $A=\nabla\log\gamma$. Define $\phi_t=\Psi_t*\phi$ and $A_t=\Psi_t*A$. Then the following results are from \cite{K} and \cite{S}.
\begin{lemma}\label{pre}
  Let $\gamma\in C^{1,\sigma}(\mathbb{R}^n)$ for $0\leq \sigma \leq 1$ and $\gamma-1\in H^{\frac{3}{2}+\sigma}(\mathbb{R}^n)$ with compact support. Then
    \begin{align*}
        &\|\nabla\cdot A_t\|_{L^\infty(\mathbb{R}^n)}\leq C t^{1-\sigma},\\
        &\|\phi_t-\phi\|_{L^\infty(\mathbb{R}^n)}\leq C t^{-1-\sigma},\\
        &\|A_t-A\|_{L^\infty(\mathbb{R}^n)}\leq C t^{-\sigma},
    \end{align*}
and
\begin{align*}
&\|\nabla\cdot A_t\|_{L^2(\mathbb{R}^n)}\leq C t^{\frac{1}{2}-\sigma},\\
&\|\phi_t-\phi\|_{L^2(\mathbb{R}^n)}\leq C t^{-\frac{3}{2}-\sigma},\\
&\|A_t-A\|_{L^2(\mathbb{R}^n)}\leq C t^{-\frac{1}{2}-\sigma}.
\end{align*}
\end{lemma}

The following lemma is taken from \cite{Z}.
\begin{lemma}[Zhang \cite{Z}]\label{2.2}
     Let $\Omega\subset \mathbb{R}^n, n\geq 2$, be a bounded domain with $C^2$ boundary and $u\in H^{1}(\Omega)$. Then there exists a constant $C$ such that
$$
     \int_{\partial\Omega}u^2dS\leq C\left\{  \left(\int_\Omega u^2dx\right)^{1/2} \left(\int_\Omega |\nabla u|^2dx\right)^{1/2}+\int_\Omega u^2 dx \right\}.
$$
\end{lemma}

We will need the stable determination of the conductivity at points on the boundary of $\Omega$. Since the stability estimate derived in \cite{A} is local, the same estimates hold for the localized Dirichlet-to-Neumann map. This result can be proved by the same arguments in \cite{A}.
\begin{theorem}\label{AL}
     Let $\gamma_j\in C^{1,\sigma}(\overline\Omega)$ satisfy $\gamma_j>\gamma_0>0$ for $j=1,2$. Then
\begin{align}\label{AL1}
       \|\gamma_1-\gamma_2\|_{L^\infty(\partial\Omega)}\lesssim \|\tilde \Lambda_{\gamma_1}-\tilde\Lambda_{\gamma_2}\|_{*}
\end{align}
and
\begin{align}\label{AL2}
       \sum_{|\alpha|=1}\|\partial^\alpha\gamma_1-\partial^\alpha\gamma_2\|_{L^\infty(\partial\Omega)}\lesssim
\|\tilde \Lambda_{\gamma_1}-\tilde\Lambda_{\gamma_2}\|^\theta_{*}
\end{align}
for some $0<\theta<1$ depending only on $\sigma$. Here the implicit constants depend on $n, \Omega, \sigma, \gamma_0$ and $\|\gamma_j\|_{C^{1, \sigma}(\overline\Omega)}$ for $j=1,2$.
\end{theorem}

We will use the following theorem to obtain the stability estimate on a large ball $B(0,R)$ by controlling an open subset of $B(0,R)$. This idea was introduced in \cite{HW}.
\begin{proposition}[Vessella \cite{V}]\label{Ve}
Let $\tau_0, d_0>0$. Let $D\subset \mathbb{R}^n$ be an open, bounded and connected set such that $\{x\in D:d(x, \partial D) >\tau\}$ is connected for any $\tau\in [0, \tau_0]$. Let $E\subset D$ be an open set such that $d(E,\partial D)\geq d_0$. If $f$ is an analytic function with
$$
      \|\partial^\alpha f\|_{L^\infty(D)}\leq \frac{M\alpha!}{\rho^{|\alpha|}},\ \ \hbox{for all $\alpha\in \mathbb{N}^n$}
$$
for some $M,\rho>0$, then
$$
|f(x)|\leq (2M)^{1-\tilde \theta(|E|/|D|)}(\|f\|_{L^\infty(E)})^{\tilde \theta(|E|/|D|)},
$$
where $\tilde \theta\in (0,1)$ depends on $d_0,\diam D, \tau_0,n
,\rho$ and $d(x,\partial D)$.
\end{proposition}

\section{Complex geometrical optics solutions}
In this section, we will review the construction of CGO solutions
for the conductivity equation following the arguments presented in
\cite{Z}, but with the conductivity in
$C^{1,\sigma}(\overline\Omega)\cap H^{\frac{3}{2}+\sigma}(\Omega),
0<\sigma<1.$ Note that the regularity assumption
$H^{\frac{3}{2}+\sigma}(\Omega)$ is used to control the $H^{1/2}$
norm of the conductivities on the boundary. The detailed discussion
will be presented in Section 4.

First, we introduce the spaces $\dot{X}^{b}_{\zeta}$ and $X^{b}_{\zeta}$ which are defined by the norm
$$
\|u\|_{\dot{X}^{b}_{\zeta}}=\|| p_\zeta(\xi)|^b \hat{u}(\xi)\|_{L^2}
$$
and
$$
\|u\|_{X^{b}_{\zeta}}=\|(|\zeta|+| p_\zeta(\xi)|)^b \hat{u}(\xi)\|_{L^2},
$$
respectively.
Here $p_\zeta(\xi)=-|\xi|^2+2i\zeta\cdot\xi$ is the symbol of $\Delta+2\zeta\cdot\nabla$.

Let $\Omega$ be an open bounded domain in $\mathbb{R}^n, n\geq 3$
with $C^2$ boundary. Let $\gamma\in C^{1,\sigma}(\overline\Omega)$
and let $u$ be the solution of $\nabla\cdot\gamma\nabla u=0$ in
$\Omega$. Then $u$ satisfies
\begin{align}\label{cgo}
     \left(-\Delta-A\cdot\nabla\right)u=0\ \ \hbox{in $\Omega$},
\end{align}
where $A=\nabla\log\gamma\in C^{0,\sigma}(\overline\Omega)$.
Suppose that the CGO solutions of (\ref{cgo}) are of the form
$$
     u=e^{-\frac{\phi_t}{2}}e^{x\cdot\zeta}\left(1+w(x,\zeta)\right),
$$
with $\phi_t=\Psi_t*\phi$ and $\zeta\in \mathbb{C}^n,\ \zeta\cdot\zeta=0$. Here we denote $\phi=\log \gamma$.
Then the function $w$ satisfies the following equation
\begin{align}
      \left(-\Delta+(A_t-A)\cdot \nabla+q_t\right) \left(e^{x\cdot\zeta}\left(1+w\right)\right)=0,
\end{align}
where $q_t=\frac{1}{2}\nabla\cdot A_t-\frac{1}{4}(A_t)^2+\frac{1}{2}A\cdot A_t$.
Equivalently, $w$ is the solution of
\begin{align}\label{solw}
      \left(-\Delta_\zeta+\left(A_t-A\right)\cdot \nabla_\zeta+q_t\right) w=  \left(A-A_t\right)\cdot \zeta-q_t,
\end{align}
where $-\Delta_\zeta=\Delta+2\zeta\cdot\nabla$ and $\nabla_\zeta=\nabla+\zeta$.

We let $\eta\in S^{n-1}$. Fix $k \in \mathbb{R}^n$ satisfying
$\eta\cdot k=0$. Let $\eta_1\in S^{n-1}$ such that $k\cdot \eta_1=
\eta\cdot \eta_1=0$. We choose
$\zeta_1=-s\eta-i\left(\frac{k}{2}-r\eta_1\right)$ and
$\zeta_2=s\eta-i\left(\frac{k}{2}+r\eta_1\right)$ such that
$|k|^2/4+r^2=s^2$, $\zeta_i\cdot \zeta_i=0$ and
$\zeta_1+\zeta_2=-ik$.

The following lemma lists some inequalities between the norms in
ordinary Sobolev spaces and the spaces $X^b_\zeta$. The inequalities
in this lemma are taken from Lemma 2.2 in \cite{HT} and Lemma 3.3 in
\cite{Z}.
\begin{lemma}\label{sobolev}
Let $\Phi_B$ be a fixed Schwartz function and write $u_B=\Phi_B u$. Then the following estimates hold:
\begin{align*}
    \|u_B\|_{L^2(\mathbb{R}^n)}&\lesssim s^{-1/2}\|u\|_{\dot{X}_\zeta^{1/2}};\,\ \
    \|u_B\|_{H^{1/2}(\mathbb{R}^n)}\lesssim \|u\|_{\dot{X}_\zeta^{1/2}};\\
    \|u_B\|_{H^1(\mathbb{R}^n)}&\lesssim s^{1/2}\|u\|_{\dot{X}_\zeta^{1/2}};\quad \ \
    \|u\|_{X_\zeta^{-1/2}}\lesssim s^{-1/2}\|u\|_{L^2(\mathbb{R}^n)}.
\end{align*}

\end{lemma}

The following result is contained in Lemma 3.4 and 3.5 in \cite{Z}.
\begin{theorem}[Zhang \cite{Z}]\label{w}
    Let $\gamma_i\in C^1(\mathbb{R}^n)$ with $\gamma_i>\gamma_0>0$ and $\gamma_i=1$ outside a ball. Then for any fixed $k\in \mathbb{R}^n$, there exists a sequence $\zeta^{(n)}_i$ with $|\zeta^{(n)}_i|=\sqrt{2}s_n$ such that
\begin{align}
     \|w_i^{(n)}\|_{\dot{X}^{1/2}_{\zeta^{(n)}_i}}\lesssim \|\left(A_{is_n}-A_i\right)\cdot\zeta^{(n)}_i+q_{s_n} \|_{\dot{X}^{-1/2}_{\zeta^{(n)}_i}} \rightarrow 0\ \ \hbox{as $s_n\rightarrow \infty$}.
\end{align}
Moreover,
$$
      \|w_i^{(n)}\|_{L^2(\Omega)}\lesssim s_n^{-1/2} \|w_i^{(n)}\|_{\dot{X}^{1/2}_{\zeta^{(n)}_i}} ;\ \ \|w_i^{(n)}\|_{H^1(\Omega)}\lesssim s_n^{1/2} \|w_i^{(n)}\|_{\dot{X}^{1/2}_{\zeta^{(n)}_i}}
$$
and
$$
  \|w_i^{(n)}\|_{H^{1/2}(\Omega)}\lesssim  \|w_i^{(n)}\|_{\dot{X}^{1/2}_{\zeta^{(n)}_i}} ;\ \ \|w_i^{(n)}\|_{H^2(\Omega)}\lesssim s_n^{3/2} \|w_i^{(n)}\|_{\dot{X}^{1/2}_{\zeta^{(n)}_i}},
$$
where $w_i^{(n)}$ is a solution of (\ref{solw}) with $t=s_n$ and
$A_i=\nabla\phi_i=\nabla\log \gamma_i$ for $i=1,2$.
\end{theorem}

From Theorem \ref{w}, we take the CGO solutions
$$
u^{(n)}_1= e^{-\frac{\phi_{1s_n}}{2}}e^{x\cdot\zeta^{(n)}_1}\left(1+w^{(n)}_1\right)
$$
and
$$
u^{(n)}_2= e^{-\frac{\phi_{2s_n}}{2}}e^{x\cdot\zeta^{(n)}_2}\left(1+w^{(n)}_2\right).
$$
The CGO solutions can also be written as
\begin{align}
u^{(n)}_i= e^{-\frac{\phi_{is_n}}{2}}e^{x\cdot\zeta^{(n)}_i}\left(1+w^{(n)}_i\right)=\sqrt{\gamma_i}^{-1}e^{x\cdot\zeta^{(n)}_i}\left(1+\psi^{(n)}_i\right)
\end{align}
for $i=1,2$. Here $\psi^{(n)}_i=\sqrt{\gamma_i}\left(e^{-\frac{\phi_{is_n}}{2}}-\sqrt{\gamma_i}^{-1}\right)+\sqrt{\gamma_i}e^{-\frac{\phi_{is_n}}{2}}w^{(n)}_i$.
For simplicity, we will not write the superscripts $(n)$ and the subscripts of $s_n$ unless otherwise particularly specified.

Note that by lemma \ref{pre} and Theorem \ref{w}, we have
\begin{align}\label{psi3.6}
      \|\psi_i\|_{L^2(\Omega )}\lesssim s^{-1-\sigma} + s^{-1/2} \|w_i\|_{\dot{X}^{1/2}_{\zeta_i}}.
\end{align}

\begin{lemma}\label{qs}
For $0<\sigma<1$, if $\lambda$ is sufficiently large we have
\begin{align}
 \frac{1}{\lambda}\int_{S^{n-1}}\int^{2\lambda}_{\lambda}  \|(A_s-A)\cdot\zeta+q_s\|^2_{\dot{X}^{-1/2}_{\zeta}}  ds d\eta\lesssim \lambda^{-2\sigma}+\lambda^{-1}.
\end{align}
\end{lemma}

\begin{proof}

     Let $\Phi$ be a cut-off function on the support of $A_s$ and $A$. Then, by Lemma 2.2 in \cite{HT} and Lemma \ref{sobolev}, we have
\begin{align*}
         &\|(A_s)^2\|^2_{\dot{X}^{-1/2}_{\zeta}}     =\|\Phi(A_s)^2\|^2_{\dot{X}^{-1/2}_{\zeta}}    \lesssim \|(A_s)^2\|^2_{X^{-1/2}_{\zeta}}    \lesssim s^{-1},    \\
         &\|A\cdot A_s\|^2_{\dot{X}^{-1/2}_{\zeta}}     =\|\Phi(A\cdot A_s)\|^2_{\dot{X}^{-1/2}_{\zeta}}    \lesssim \|A\cdot A_s\|^2_{X^{-1/2}_{\zeta}}    \lesssim s^{-1}.
\end{align*}

Observing that $|(\nabla\cdot A_s)\hat(\xi)|=|\xi\cdot \hat{A_s}|=|\xi\cdot\hat\Psi(\frac{\xi}{s})\hat A(\xi)|\leq \|\hat\Psi(\frac{\xi}{s})\|_{L^\infty(\mathbb{R}^n)}|\xi\cdot \hat{A}|\lesssim |(\nabla\cdot A)\hat(\xi)|$. Then $\|\nabla\cdot A_s\|^2_{\dot{X}^{-1/2}_{\zeta}}\lesssim  \|\nabla\cdot A\|^2_{\dot{X}^{-1/2}_{\zeta}} $. Let $h=\sqrt{\lambda}$ and $\Psi_h=h^n\Psi(hx)$ as in Lemma \ref{pre}, we have
\begin{align}\label{3.8}
       \frac{1}{\lambda}\int_{S^{n-1}}\int^{2\lambda}_{\lambda}  \|\nabla\cdot A_s\|^2_{\dot{X}^{-1/2}_{\zeta}} ds d\eta
&\lesssim    \frac{1}{\lambda}\int_{S^{n-1}}\int^{2\lambda}_{\lambda}  \|\nabla\cdot A\|^2_{\dot{X}^{-1/2}_{\zeta}} ds d\eta   \notag\\
&\lesssim     \frac{1}{\lambda}\int_{S^{n-1}}\int^{2\lambda}_{\lambda}  \|\nabla\cdot (\Psi_h*A)\|^2_{\dot{X}^{-1/2}_{\zeta}} ds d\eta\notag\\
&\quad + \frac{1}{\lambda}\int_{S^{n-1}}\int^{2\lambda}_{\lambda}  \|\nabla\cdot (\Psi_h*A-A)\|^2_{\dot{X}^{-1/2}_{\zeta}} ds d\eta.
\end{align}
Using Lemma 3.1 in \cite{HT} and Lemma \ref{pre}, (\ref{3.8})
follows that
\begin{align}\label{3.9}
        \frac{1}{\lambda}\int_{S^{n-1}}\int^{2\lambda}_{\lambda}  \|\nabla\cdot A_s\|^2_{\dot{X}^{-1/2}_{\zeta}} ds d\eta
&\lesssim    \frac{1}{\lambda}  \|\nabla\cdot (\Psi_h*A)\|^2_{L^2(\mathbb{R}^n)}+  \|\Psi_h*A-A\|^2_{L^2(\mathbb{R}^n)} \notag\\
&\lesssim    \frac{1}{\lambda} h^{1-2\sigma}+h^{-1-2\sigma}  \lesssim    \lambda^{-\frac{1}{2}-\sigma}.
\end{align}

By the definition of $q_s$, we can deduce that
\begin{align}
             &\frac{1}{\lambda}\int_{S^{n-1}}\int^{2\lambda}_{\lambda}  \|q_s\|^2_{\dot{X}^{-1/2}_{\zeta}}  ds d\eta  \notag\\
&\lesssim   \frac{1}{\lambda}\int_{S^{n-1}}\int^{2\lambda}_{\lambda}  \|\nabla\cdot A_s\|^2_{\dot{X}^{-1/2}_{\zeta}} +\|(A_s)^2\|^2_{\dot{X}^{-1/2}_{\zeta}}+\|A\cdot A_s\|^2_{\dot{X}^{-1/2}_{\zeta}} ds d\eta  \notag\\
&\lesssim   \lambda^{-\frac{1}{2}-\sigma}+\lambda^{-1}.
\end{align}
Applying Lemma 2.2 in \cite{HT} and Lemma \ref{sobolev}, we get
$$
 \|(A_s-A)\cdot\zeta\|^2_{\dot{X}^{-1/2}_{\zeta}} \lesssim  s^2 \|\Phi (A_s-A)\|^2_{\dot{X}^{-1/2}_{\zeta}} \lesssim  s^2 \|A_s-A\|^2_{X^{-1/2}_{\zeta}}\lesssim s \|A_s-A\|^2_{L^2(\mathbb{R}^n)}.
 $$
Thus we derive
\begin{align}
             \frac{1}{\lambda}\int_{S^{n-1}}\int^{2\lambda}_{\lambda}  \|(A_s-A)\cdot\zeta\|^2_{\dot{X}^{-1/2}_{\zeta}}  ds d\eta  \lesssim  \lambda^{-2\sigma}
\end{align}
from Lemma \ref{pre}.
The proof is completed.
\end{proof}

Note that  $ \|w\|^2_{\dot{X}^{1/2}_{\zeta}}\lesssim \|(A_s-A)\cdot\zeta+q_s\|^2_{\dot{X}^{-1/2}_{\zeta}}$. By lemma \ref{qs}, we obtain the following estimate
\begin{align}\label{normw}
 \frac{1}{\lambda}\int_{S^{n-1}}\int^{2\lambda}_{\lambda}  \|w\|^2_{\dot{X}^{1/2}_{\zeta}}  ds d\eta\lesssim \lambda^{-2\sigma}+\lambda^{-1}.
\end{align}

The following Carleman estimate is deduced by Zhang by using the Carleman estimate in the paper \cite{K}.
\begin{theorem}[Zhang \cite{Z}]\label{car}
     Let $\eta\in S^{n-1}$ and $u\in H^2(\Omega)$. Suppose that $\gamma\in C^1(\Omega)$. Then there exists a constant $s_0>0$ such that for $s\geq s_0$, we have
\begin{align}\label{carleman}
     &C\left(s^2\|u\|_{L^2(\Omega)}^2+\|\nabla u\|_{L^2(\Omega)}^2\right)-C_1s^2\int_{\partial\Omega}|u|^2 dS \notag\\
&-C_2\int_{\partial\Omega}\overline u\partial_\nu u dS+\int_{\partial\Omega} 4s \Re{\left(\partial_\nu u\partial_\eta \overline u\right)}-2s(\nu\cdot \eta)|\nabla u|^2+2s^3(\nu\cdot \eta)|u|^2 dS \notag\\
&\leq \|e^{-x\cdot s\eta}\left(-\Delta+\left(A_s-A\right)\cdot \nabla+q_s\right)\left(e^{x\cdot s\eta}u\right)\|_{L^2(\Omega)}^2.
\end{align}
\end{theorem}

We also need the following result.
\begin{proposition}[Knudsen \cite{K}]\label{Knu}
     Suppose $\gamma_j\in C^1(\overline\Omega)$ and $u_j\in H^1(\Omega)$ satisfy $\nabla\cdot \gamma_j\nabla u_j=0$ in $\Omega$ for $j=1,2$. Suppose that $\tilde u_1\in H^1(\Omega)$ satisfies $\nabla\cdot \gamma_1\nabla \tilde u_1=0$ with $\tilde u_1=u_2$ on $\partial\Omega$. Then
\begin{align}\label{boundaryint}
     &\int_\Omega \left(\sqrt{\gamma_1}\nabla \sqrt{\gamma_2}-\sqrt{\gamma_2}\nabla \sqrt{\gamma_1}  \right) \cdot \nabla\left(u_1u_2\right) dx\notag\\
     &=\int_{\partial\Omega} \gamma_1\partial_\nu\left(\tilde u_1-u_2\right)u_1 dS+\int_{\partial\Omega} (\gamma_1-\sqrt{\gamma_1\gamma_2})(u_1\partial_\nu u_2-u_2\partial_\nu u_1)dS,
\end{align}
where the integral is understood in the sense of the dual pairing between $H^{1/2}(\partial\Omega)$ and $H^{-1/2}(\partial\Omega)$.
\end{proposition}
Note that this proposition is slightly different from the Lemma 4.1
in \cite{K} due to different assumptions on
$\gamma|_{\partial\Omega}$. In \cite{K}, they have
$\gamma_1=\gamma_2$ on $\partial\Omega$, so the second term on the
right hand side of (\ref{boundaryint}) vanishes.

Using Theorem \ref{AL} and the trace theorem, we get
\begin{align}\label{3.15inq}
  \left| \int_{\partial\Omega_{-,\varepsilon}}((\gamma_1-\gamma_2)\partial_\nu u_2)  u_1 dS \right|^2
&\lesssim    \|\gamma_1-\gamma_2\|^2_{L^\infty(\partial\Omega)} \|\nabla u_2\|^2_{H^1(\Omega)} \|u_1\|^2_{H^1(\Omega)} \notag\\
&\lesssim     \|\tilde \Lambda_{\gamma_1}-\tilde\Lambda_{\gamma_2}\|^{2}_* \|u_2\|^2_{H^2(\Omega)} \|u_1\|^2_{H^1(\Omega)}.
\end{align}
Note that since $\gamma_2\in C^{1,\sigma}$, the elliptic regularity
theorem implies that $u_2\in H^2(\Omega).$ By using the equality
that
$$
     \gamma_1\partial_\nu(\tilde u_1-u_2)u_1=(\gamma_1\partial_\nu\tilde u_1-\gamma_2\partial_\nu u_2)u_1+((\gamma_1-\gamma_2)\partial_\nu u_2)u_1
$$
and (\ref{3.15inq}), we have
\begin{align}\label{r1}
      \left| \int_{\partial\Omega_{-,\varepsilon}}  \gamma_1\partial_\nu(\tilde u_1-u_2)u_1 dS  \right|^2
\lesssim     \|\tilde \Lambda_{\gamma_1}-\tilde\Lambda_{\gamma_2}\|^{2}_*  \|u_2\|^2_{H^2(\Omega)} \|u_1\|^2_{H^1(\Omega)}.
\end{align}
Proposition \ref{Knu} and (\ref{r1}) imply that
\begin{align}\label{int}
     &\left|\int_\Omega \left(\sqrt{\gamma_1}\nabla \sqrt{\gamma_2}-\sqrt{\gamma_2}\nabla \sqrt{\gamma_1} \right) \cdot \nabla\left(u_1u_2\right) dx\right|^2\notag\\
&\lesssim  \left|\int_{\partial\Omega_{+,\varepsilon}} \gamma_1\partial_\nu\left(\tilde u_1-u_2\right)u_1 dS\right|^2+ \|\tilde \Lambda_{\gamma_1}-\tilde\Lambda_{\gamma_2}\|^{2}_* \|u_2\|^2_{H^2(\Omega)} \|u_1\|^2_{H^2(\Omega)}.
\end{align}

In the remaining part of this section, we will estimate the first
term on the right hand side of (\ref{int}). Denote
$u_0=e^{\frac{\phi_{1s}}{2}}\left(\tilde u_1-u_2\right)$ and $\delta
u=\left(e^{\frac{\phi_{1s}}{2}}-e^{\frac{\phi_{2s}}{2}}\right)u_2$.
Let $u=u_0+\delta u$. Observing that
\begin{align}\label{3.16}
           \left|\int_{\partial\Omega_{+,\varepsilon}} \gamma_1\partial_\nu(\tilde u_1-u_2)u_1 dS\right|^2
&\lesssim      \| 1+w_1 \|^2_{L^2(\partial\Omega_{+,\varepsilon})} \int_{\partial\Omega_{+,\varepsilon}} e^{-2x\cdot
            s\eta} |\partial_\nu(\tilde u_1-u_2)|^2 dS \notag\\
&\lesssim    \left(\int_{\partial\Omega_{+,\varepsilon}}
            e^{-2x\cdot s\eta} |\partial_\nu u|^2 dS+\int_{\partial\Omega_{+,\varepsilon}} e^{-2x\cdot s\eta}
            |\partial_\nu \delta u|^2 dS\right),
\end{align}
Here we use the face that if $s$ is large,
$\|w_1\|^2_{\dot{X}_\zeta^{1/2}}$ is small compared to $1$ according
to Theorem \ref{w}. Thus
 \begin{align*}
 \| 1+w_1 \|^2_{L^2(\partial\Omega_{+,\varepsilon})} \lesssim 1+\|w_1\|^2_{\dot{X}_\zeta^{1/2}}\lesssim 1
 \end{align*}
by applying Lemma \ref{2.2}.

\begin{lemma}\label{ldelta}
       Let $\Omega\subset \mathbb{R}^n, n\geq 3$, be an open and bounded domain with $C^2$ boundary.
       For $i=1,2$, let $\gamma_i\in C^{1,\sigma}(\overline\Omega)\cap H^{\frac{3}{2}+\sigma}(\Omega)$
       be a real-valued function and $\gamma_i>\gamma_0>0$. Suppose that
       $\gamma_1|_{\partial\Omega_{+,\varepsilon}}=\gamma_2|_{\partial\Omega_{+,\varepsilon}}$
       and $\partial_\nu \gamma_1|_{\partial\Omega_{+,\varepsilon}}=\partial_\nu \gamma_2|_{\partial\Omega_{+,\varepsilon}}$.
       If $s$ is large, then
\begin{align}\label{delta}
                & \int_{\partial\Omega_{-,\varepsilon}}  e^{-2x\cdot s\eta}|\nabla \delta u|^2 dS
 \lesssim       \left(s^{-2\sigma}+ \|\tilde \Lambda_{\gamma_1}-\tilde\Lambda_{\gamma_2}\|^{2\theta}_{*}+s^2\|\tilde \Lambda_{\gamma_1}-\tilde\Lambda_{\gamma_2}\|^2_{*}  \right),\\
                & \int_{\partial\Omega_{-,\varepsilon}}  e^{-2x\cdot s\eta} |\delta u|^2 dS
                \lesssim        \left(s^{-2-2\sigma}+\|\tilde \Lambda_{\gamma_1}-\tilde\Lambda_{\gamma_2}\|^2_{*}\right)
.
 \end{align}
Moreover, we have
 \begin{align}\label{delta2}
                 & \int_{\partial\Omega_{+,\varepsilon}}  e^{-2x\cdot s\eta}|\nabla \delta u|^2 dS
 \lesssim       s^{-2\sigma},\\
                & \int_{\partial\Omega_{+,\varepsilon}}  e^{-2x\cdot s\eta} |\delta u|^2 dS
 \lesssim        s^{-2-2\sigma}
\end{align}
when $s$ is sufficiently large.
\end{lemma}

\begin{proof}
We will prove the estimate for
$\int_{\partial\Omega_{-,\varepsilon}}  e^{-2x\cdot s\eta}|\nabla
\delta u|^2 dS$ first. We consider
\begin{align}\label{1}
                  \int_{\partial\Omega_{-,\varepsilon}}  e^{-2x\cdot s\eta}|\nabla \delta u|^2 dS
& \lesssim         \int_{\partial\Omega_{-,\varepsilon}}  e^{-2x\cdot s\eta}\left| \nabla \left( e^{\frac{\phi_{1s}}{2}}-e^{\frac{\phi_{2s}}{2}}  \right)  \right|^2 |u_2|^2 dS\notag\\
    &\quad+  \int_{\partial\Omega_{-,\varepsilon}}  e^{-2x\cdot s\eta}\left|e^{\frac{\phi_{1s}}{2}}-e^{\frac{\phi_{2s}}{2}} \right|^2 |\nabla u_2|^2 dS.
\end{align}
Using Theorem \ref{AL} and Lemma \ref{pre}, the first term of the right side of (\ref{1}) can be written as
\begin{align*}
              &\int_{\partial\Omega_{-,\varepsilon}}  e^{-2x\cdot s\eta}\left| \nabla \left( e^{\frac{\phi_{1s}}{2}}-e^{\frac{\phi_{2s}}{2}} \right)  \right|^2 |u_2|^2 dS\\
 &\lesssim   \int_{\partial\Omega_{-,\varepsilon}}  e^{-2x\cdot s\eta}\left(\left| \nabla \left( e^{\frac{\phi_{1s}}{2}}-\sqrt{\gamma_1}  \right)  \right|^2+\left|\nabla \left(\sqrt{\gamma_1}-\sqrt{\gamma_2}\right)\right|^2 +   \left| \nabla \left( e^{\frac{\phi_{2s}}{2}}-\sqrt{\gamma_2}  \right)  \right|^2\right) |u_2|^2 dS\\
 &\lesssim   \sum_{j=1}^2 \left(\|A_{js}-A_j\|^2_{L^\infty(\overline\Omega)}+  \|  \nabla (\gamma_1-\gamma_2) \|^2_{L^\infty(\partial\Omega)} + \| \gamma_1-\gamma_2 \|^2_{L^\infty(\partial\Omega)}   +\|\phi_{js}-\phi_j\|^2_{L^\infty(\overline\Omega)}\right) \\
 &\quad \left( \| 1+w_2 \|^2_{L^2(\partial\Omega)}\right)\\
&\lesssim    \left(s^{-2\sigma}+s^{-2-2\sigma}+\|\tilde \Lambda_{\gamma_1}-\tilde\Lambda_{\gamma_2}\|^{2\theta}_{*}+\|\tilde \Lambda_{\gamma_1}-\tilde\Lambda_{\gamma_2}\|^{2}_{*}\right)\left(1+\|w_2\|^2_{\dot{X}^{1/2}_{\zeta_2}} \right).
\end{align*}
We use similar arguments to estimate the second term of (\ref{1}).
\begin{align*}
            &\int_{\partial\Omega_{-,\varepsilon}}  e^{-2x\cdot s\eta}\left|e^{\frac{\phi_{1s}}{2}}-e^{\frac{\phi_{2s}}{2}} \right|^2 |\nabla u_2|^2 dS\\
&\lesssim    \int_{\partial\Omega_{-,\varepsilon}}\left(\left|  e^{\frac{\phi_{1s}}{2}}-\sqrt{\gamma_1} \right|^2+  \left|\sqrt{\gamma_1}-\sqrt{\gamma_2}\right|^2   +\left| e^{\frac{\phi_{2s}}{2}}-\sqrt{\gamma_2}\right|^2\right) |\nabla w_2|^2 dS   \\
 &\quad +s^2\int_{\partial\Omega_{-,\varepsilon}}\left(\left|  e^{\frac{\phi_{1s}}{2}}-\sqrt{\gamma_1} \right|^2+  \left|\sqrt{\gamma_1}-\sqrt{\gamma_2}\right|^2+\left| e^{\frac{\phi_{2s}}{2}}-\sqrt{\gamma_2}\right|^2\right) \left(1+|w_2|^2\right) dS   \\
&\lesssim   \sum_{j=1}^2 \left(\|\phi_{js}-\phi_j\|^2_{L^\infty(\overline\Omega)}+  \|\sqrt{\gamma_1}-\sqrt{\gamma_2}\|^2_{L^\infty(\partial\Omega)}\right)\left(\|\nabla w_2\|^2_{L^2(\partial \Omega)}  +s^2 \left(1+ \| w_2\|^2_{L^2(\partial \Omega)}\right)\right)\\
&\lesssim     \left(s^{-2\sigma}+s^2\|\tilde \Lambda_{\gamma_1}-\tilde\Lambda_{\gamma_2}\|^2_{*}\right)\left(1+\|w_2\|^2_{\dot{X}^{1/2}_{\zeta_2}} \right).
\end{align*}
Thus we have
$$
        \int_{\partial\Omega_{-,\varepsilon}}  e^{-2x\cdot s\eta}|\nabla  \delta u|^2 dS
 \lesssim       \left(s^{-2\sigma}+ \|\tilde \Lambda_{\gamma_1}-\tilde\Lambda_{\gamma_2}\|^{2\theta}_{*}+s^2\|\tilde \Lambda_{\gamma_1}-\tilde\Lambda_{\gamma_2}\|^2_{*}  \right)\left(1+\|w_2\|^2_{\dot{X}^{1/2}_{\zeta_2}} \right).
$$

Since $\gamma_1|_{\partial\Omega_{+,\varepsilon}}=\gamma_2|_{\partial\Omega_{+,\varepsilon}}$ and $\partial_\nu \gamma_1|_{\partial\Omega_{+,\varepsilon}}=\partial_\nu \gamma_2|_{\partial\Omega_{+,\varepsilon}}$, the estimate of
$\int_{\partial\Omega_{+,\varepsilon}}  e^{-2x\cdot s\eta}|\nabla \delta u|^2 dS$ does not contain the $\|\tilde \Lambda_{\gamma_1}-\tilde\Lambda_{\gamma_2}\|_{*} $ terms. Thus,
$$
       \int_{\partial\Omega_{+,\varepsilon}}  e^{-2x\cdot s\eta}|\nabla \delta u|^2 dS
 \lesssim       s^{-2\sigma}\left(1+\|w_2\|^2_{\dot{X}^{1/2}_{\zeta_2}} \right).
$$

Similarly, we can deduce that
\begin{align*}
                  &\int_{\partial\Omega_{-,\varepsilon}}  e^{-2x\cdot s\eta}|\delta u|^2 dS   \lesssim     \left(s^{-2-2\sigma}+\|\tilde \Lambda_{\gamma_1}-\tilde\Lambda_{\gamma_2}\|^2_{*}\right)\left(1+\|w_2\|^2_{\dot{X}^{1/2}_{\zeta_2}} \right)
\end{align*}
and
\begin{align*}
                  &\int_{\partial\Omega_{+,\varepsilon}}  e^{-2x\cdot s\eta}|\delta u|^2 dS      \lesssim     s^{-2-2\sigma}\left(1+\|w_2\|^2_{\dot{X}^{1/2}_{\zeta_2}} \right).
\end{align*}
Since $\|w_2\|^2_{\dot{X}^{1/2}_{\zeta_2}}$ is small compared to $1$ when $s$ is large, we complete the proof.
\end{proof}

\begin{lemma}\label{lemma3.6}
    Under the same assumption as Lemma \ref{ldelta}, we have
\begin{align}\label{delta}
                  \int_{\partial\Omega_{+,\varepsilon}}  e^{-2x\cdot s\eta}|\partial_\nu u|^2 dS
 &\lesssim      s^{-2\sigma}+s^{-1}+ \|\tilde \Lambda_{\gamma_1}-\tilde\Lambda_{\gamma_2}\|^{2\theta}_{*}+s^2\|\tilde \Lambda_{\gamma_1}-\tilde\Lambda_{\gamma_2}\|^2_{*}   \notag\\
              &\quad  +e^{cs}( \|\tilde\Lambda_{\gamma_1}-\tilde\Lambda_{\gamma_2}\|_{*}+ \|\tilde\Lambda_{\gamma_1}-\tilde\Lambda_{\gamma_2}\|^2_{*}  )\|u_2\|^2_{H^2(\Omega)}
\end{align}
for some $0<\theta<1$ when $s$ is sufficiently large.
\end{lemma}

\begin{proof}
Since $\gamma_1>\gamma_0>0$, we have
$$
     |\partial_\nu (\tilde u_1-u_2)|^2\leq|\gamma_1 \partial_\nu \tilde u_1- \gamma_2\partial_\nu u_2|^2+
|(\gamma_1-\gamma_2)\partial_\nu u_2|^.
$$
The interpolation theory implies that
\begin{align*}
\|\gamma_1 \partial_\nu \tilde u_1-
\gamma_2\partial_\nu u_2\|^2_{L^2(\partial\Omega_{-,\varepsilon})}
 & \lesssim \|\gamma_1 \partial_\nu \tilde u_1- \gamma_2\partial_\nu u_2\|_{H^{1/2}(\partial\Omega_{-,\varepsilon})} \|(\tilde\Lambda_{\gamma_1}-\tilde\Lambda_{\gamma_2})u_2\|_{H^{-1/2}(\partial\Omega_{-,\varepsilon})}\\
 & \lesssim (\|\tilde u_1\|_{H^2(\Omega)}+\|u_2\|_{H^2(\Omega)}) \|\tilde\Lambda_{\gamma_1}-\tilde\Lambda_{\gamma_2}\|_{*} \|u_2\|_{H^1(\Omega)}.
\end{align*}
Thus we can deduce
\begin{align*}
     \|\partial_\nu (\tilde u_1-u_2)\|^2_{L^2(\partial\Omega_{-,\varepsilon})}
     &\lesssim (\|\tilde u_1\|_{H^2(\Omega)}+\|u_2\|_{H^2(\Omega)}) \|\tilde\Lambda_{\gamma_1}-\tilde\Lambda_{\gamma_2}\|_{*} \|u_2\|_{H^1(\Omega)} \\
      & \quad+ \|\tilde\Lambda_{\gamma_1}-\tilde\Lambda_{\gamma_2}\|^2_{*}  \|u_2\|^2_{H^2(\Omega)}
\end{align*}
from Theorem \ref{AL}. By elliptic regularity theorem and $\tilde
u_1|_{\partial\Omega}=u_2|_{\partial\Omega}$, $\|\tilde
u_1\|_{H^2(\Omega)}\lesssim \|u_2\|_{H^2(\Omega)}$. Thus we have
\begin{align}\label{u0}
          \int_{\partial\Omega_{-,\varepsilon}} e^{-2x\cdot s\eta}|\partial_\nu u_0|^2 dS
      &\lesssim  \int_{\partial\Omega_{-,\varepsilon}}   e^{-2x\cdot s\eta}\left|\partial_\nu (\tilde u_1-u_2)\right|^2 dS  \notag\\
      &\lesssim  e^{cs}( \|\tilde\Lambda_{\gamma_1}-\tilde\Lambda_{\gamma_2}\|_{*}+ \|\tilde\Lambda_{\gamma_1}-\tilde\Lambda_{\gamma_2}\|^2_{*}  )\|u_2\|^2_{H^2(\Omega)}
\end{align}
by using the fact that $u_0|_{\partial\Omega}=0$.
Let $v=e^{-x\cdot s\eta }u$. We plug $v$ into the Carleman estimate in Theorem \ref{car}, then we get that
\begin{align*}
   \int_{\partial\Omega_{+,\varepsilon}} 4 \Re{(\partial_\nu v\partial_\eta \overline v)}-2(\nu\cdot \eta)|\nabla v|^2 dS
&\lesssim  s\int_{\partial\Omega}|v|^2 dS
+   \int_{\partial\Omega}  s^2(\nu\cdot \eta)|v|^2 dS+\frac{1}{s}\int_{\partial\Omega}\overline v\partial_\nu v dS \\
&\quad+\frac{1}{s} \|e^{-x\cdot s\eta}(-\Delta+(A_{1s}-A_1)\cdot \nabla+q_{1s})(e^{x\cdot s\eta}v)\|_{L^2(\Omega)}^2\\
&\quad+\int_{\partial\Omega_{-,\varepsilon}} 4 \Re{(\partial_\nu v\partial_\eta \overline v)}-2(\nu\cdot \eta)|\nabla v|^2 dS\\
&=:I+II+III+IV+V.
\end{align*}
For $I$ and $II$, since $u_0|_{\partial\Omega}=0$, it follows that
\begin{align*}
s\int_{\partial\Omega}|v|^2 dS
\lesssim   s\int_{\partial\Omega} e^{-2x\cdot s\eta} |\delta u|^2 dS
\lesssim   \left(s^{-1-2\sigma}+s\|\tilde \Lambda_{\gamma_1}-\tilde\Lambda_{\gamma_2}\|^2_{*}\right)
\end{align*}
and
\begin{align*}
s^2\int_{\partial\Omega}(\nu\cdot \eta)|v|^2 dS
\lesssim    s^2\int_{\partial\Omega} e^{-2x\cdot s\eta} |\delta u|^2 dS
\lesssim   \left(s^{-2\sigma}+s^2\|\tilde \Lambda_{\gamma_1}-\tilde\Lambda_{\gamma_2}\|^2_{*}\right).
\end{align*}
To estimate $III$,
first we observe that
\begin{align*}
       \frac{1}{s}\int_{\partial\Omega}   e^{-2x\cdot s\eta} \overline{\delta u} \partial_\nu u dS
&\lesssim    \frac{1}{s}\int_{\partial\Omega}   e^{-2x\cdot s\eta} |\delta u|^2dS+\frac{1}{s}\int_{\partial\Omega_{-,\varepsilon}}   e^{-2x\cdot s\eta} |\partial_\nu \delta u|^2dS\\
&\quad+\frac{1}{s}\int_{\partial\Omega_{-,\varepsilon}}   e^{-2x\cdot s\eta} |\partial_\nu u_0|^2dS +\frac{1}{s}\int_{\partial\Omega_{+,\varepsilon}}   e^{-2x\cdot s\eta} |\partial_\nu u|^2dS\\
&\lesssim      s^{-1}\left(s^{-2\sigma}+ \|\tilde \Lambda_{\gamma_1}-\tilde\Lambda_{\gamma_2}\|^{2\theta}_{*}+s^2\|\tilde \Lambda_{\gamma_1}-\tilde\Lambda_{\gamma_2}\|^2_{*}  \right)\\
&\quad+\frac{1}{s}\int_{\partial\Omega_{-,\varepsilon}}   e^{-2x\cdot s\eta} |\partial_\nu u_0|^2dS +\frac{1}{s}\int_{\partial\Omega_{+,\varepsilon}}   e^{-2x\cdot s\eta} |\partial_\nu u|^2dS.
\end{align*}
Since $u_0|_{\partial\Omega}=0$, we derive that
\begin{align*}
        III
&=    \frac{1}{s}\int_{\partial\Omega}   e^{-x\cdot s\eta} \overline{\delta u} \partial_\nu (e^{-x\cdot s\eta}   u) dS\\
&=    -(\nu\cdot \eta)\int_{\partial\Omega}   e^{-2x\cdot s\eta} |\delta u|^2 dS  +\frac{1}{s}\int_{\partial\Omega}   e^{-2x\cdot s\eta} \overline{\delta u} \partial_\nu u dS\\
&\lesssim      s^{-1}\left(s^{-2\sigma}+ \|\tilde \Lambda_{\gamma_1}-\tilde\Lambda_{\gamma_2}\|^{2\theta}_{*}+s^2\|\tilde \Lambda_{\gamma_1}-\tilde\Lambda_{\gamma_2}\|^2_{*}  \right)\\
&\quad+\frac{1}{s}\int_{\partial\Omega_{-,\varepsilon}}   e^{-2x\cdot s\eta} |\partial_\nu u_0|^2dS +\frac{1}{s}\int_{\partial\Omega_{+,\varepsilon}}   e^{-2x\cdot s\eta} |\partial_\nu u|^2dS.
\end{align*}

Next we estimate $IV$,
\begin{align*}
          IV
&\leq     \frac{1}{s} \int_{\Omega}  e^{-2x\cdot s\eta}|(-\Delta+(A_{1s}-A_1)\cdot \nabla+q_{1s})e^{x\cdot \zeta_2}(1+w_2)|^2  dx\\
 & \leq     \frac{1}{s}  \int_{\Omega}   e^{-2x\cdot s\eta}|((A_{1s}-A_1)-(A_{2s}-A_2))\cdot \nabla(e^{x\cdot \zeta_2}(1+w_2))  \\
 &\quad  +(q_{1s}-q_{2s}) e^{x\cdot \zeta_2}(1+w_2)|^2    dx.
\end{align*}
Then we deduce that
\begin{align*}
          IV
&\lesssim    s  \int_\Omega   \sum_{j=1}^2 |A_{js}-A_j|^2dx +s \int_\Omega  \sum_{j=1}^2 |A_{js}-A_j|^2|w_2|^2dx
+\frac{1}{s}  \int_\Omega  \sum_{j=1}^2 |A_{js}-A_j|^2|\nabla w_2|^2dx \\
&\quad+ \frac{1}{s} \int_\Omega |q_{2s}-q_{1s}|^2 dx
+\frac{1}{s} \int_\Omega |q_{2s}-q_{1s}|^2 |w_2|^2 dx \\
&\lesssim     s \sum_{j=1}^2 \|A_{js}-A_j\|^2_{L^2} +\sum_{j=1}^2 \|A_{js}-A_j\|^2_{L^\infty} \|w_2\|^2_{\dot{X}^{1/2}_{\zeta_2}}+  \frac{1}{s}\|q_{2s}-q_{1s}\|^2_{L^2} \\
&\quad+\frac{1}{s^2}\|q_{2s}-q_{1s}\|^2_{L^\infty} \|w_2\|^2_{\dot{X}^{1/2}_{\zeta_2}}\\
&\lesssim          s^{-2\sigma}+s^{-1}+   s^{-2\sigma}\|w_2\|^2_{\dot{X}^{1/2}_{\zeta_2}}
\end{align*}
from Lemma \ref{pre}.

Finally, for $V$, since $u_0|_{\partial\Omega}=0$ implies that $\nabla u_0=\partial_\nu u_0$ on $\partial\Omega$. Then we have
\begin{align*}
&\left|\int_{\partial\Omega_{-,\varepsilon}} 4 \Re{(\partial_\nu v\partial_\eta \overline v)}-2(\nu\cdot \eta)|\nabla v|^2 dS\right|\\
&\lesssim     \int_{\partial\Omega_{-,\varepsilon}}|\nabla v|^2 dS\\
&\lesssim     s^2\int_{\partial\Omega_{-,\varepsilon}}e^{-2x\cdot s\eta}|\delta u|^2  dS+\int_{\partial\Omega_{-,\varepsilon}}e^{-2x\cdot s\eta}|\nabla\delta u|^2  dS+\int_{\partial\Omega_{-,\varepsilon}}e^{-2x\cdot s\eta}|\partial_\nu u_0|^2  dS\\
&\lesssim      \left(s^{-2\sigma}+ \|\tilde \Lambda_{\gamma_1}-\tilde\Lambda_{\gamma_2}\|^{2\theta}_{*}+s^2\|\tilde \Lambda_{\gamma_1}-\tilde\Lambda_{\gamma_2}\|^2_{*}  \right)+\int_{\partial\Omega_{-,\varepsilon}}e^{-2x\cdot s\eta}|\partial_\nu u_0|^2  dS.
\end{align*}
Combining the estimates from $I$ to $V$, we obtain
\begin{align}\label{3}
   &\int_{\partial\Omega_{+,\varepsilon}} 4 \Re{(\partial_\nu v\partial_\eta \overline v)}-2(\nu\cdot \eta)|\nabla v|^2 dS \notag \\
&\lesssim    s^{-2\sigma}+s^{-1}+
 \|\tilde \Lambda_{\gamma_1}-\tilde\Lambda_{\gamma_2}\|^{2\theta}_{*}+s^2\|\tilde \Lambda_{\gamma_1}-\tilde\Lambda_{\gamma_2}\|^2_{*}     \notag\\
            &\quad+\int_{\partial\Omega_{-,\varepsilon}}e^{-2x\cdot s\eta}|\partial_\nu u_0|^2  dS+ \frac{1}{s}\int_{\partial\Omega_{+,\varepsilon}}   e^{-2x\cdot s\eta} |\partial_\nu u|^2dS
\end{align}
since $\|w_2\|^2_{\dot{X}^{1/2}_{\zeta_2}}$ can be neglected when $s$ is sufficiently large.

Moreover, for $(\nu\cdot \eta)>\varepsilon >0$, we have
\begin{align}\label{4}
              &\int_{\partial\Omega_{+,\varepsilon}} 4 \Re{(\partial_\nu v\partial_\eta \overline v)}-2(\nu\cdot \eta)|\nabla v|^2 dS \notag\\
&\geq
             \int_{\partial\Omega_{+,\varepsilon}} (\nu\cdot \eta)  e^{-2x\cdot s\eta} |\partial_\nu u|^2dS-
               s^2\int_{\partial\Omega_{+,\varepsilon}}e^{-2x\cdot s\eta}|\delta u|^2  dS-\int_{\partial\Omega_{+,\varepsilon}}e^{-2x\cdot s\eta}|\nabla\delta u|^2.
\end{align}
Combining (\ref{u0}), (\ref{3}) and (\ref{4}) and Lemma \ref{ldelta}, the proof is completed.
\end{proof}

From (\ref{3.16}), Lemma \ref{ldelta} and Lemma \ref{lemma3.6}, we can deduce
\begin{align}\label{3.28}
           \left|\int_{\partial\Omega_{+,\varepsilon}} \gamma_1\partial_\nu(\tilde u_1-u_2)u_1 dS\right|^2
&\lesssim      s^{-2\sigma}+s^{-1}+ \|\tilde \Lambda_{\gamma_1}-\tilde\Lambda_{\gamma_2}\|^{2\theta}_{*}+s^2\|\tilde \Lambda_{\gamma_1}-\tilde\Lambda_{\gamma_2}\|^2_{*}   \notag\\
              &\quad  +e^{cs}( \|\tilde\Lambda_{\gamma_1}-\tilde\Lambda_{\gamma_2}\|_{*}+ \|\tilde\Lambda_{\gamma_1}-\tilde\Lambda_{\gamma_2}\|^2_{*}  )\|u_2\|^2_{H^2(\Omega)}.
\end{align}
Note that $\|u_2\|^2_{H^2(\Omega)}\lesssim e^{cs}$ and
$\|u_1\|^2_{H^2(\Omega)}\lesssim e^{cs}$. Therefore,
\begin{align}\label{int2}
     &\left|\int_\Omega \left(\sqrt{\gamma_1}\nabla \sqrt{\gamma_2}-\sqrt{\gamma_2}\nabla \sqrt{\gamma_1} \right) \cdot \nabla\left(u_1u_2\right) dx\right|^2\notag\\
&\lesssim      s^{-2\sigma}+s^{-1}+ \|\tilde \Lambda_{\gamma_1}-\tilde\Lambda_{\gamma_2}\|^{2\theta}_{*}
             +e^{cs}( \|\tilde\Lambda_{\gamma_1}-\tilde\Lambda_{\gamma_2}\|_{*}+ \|\tilde\Lambda_{\gamma_1}-\tilde\Lambda_{\gamma_2}\|^2_{*}  ).
\end{align}
from (\ref{int}) and (\ref{3.28}).

\section{Stability result}
We consider the function $v:=\log \sqrt{\gamma_1}-\log\sqrt{\gamma_2}\in H^1(\Omega)$. This function $v$ is a weak solution of
\begin{align}\label{equF}
     \Delta v+\nabla(\log \sqrt{\gamma_1}+\log\sqrt{\gamma_2})\nabla v&=F \ \ \mbox{in $\Omega$}\\
      v|_{\partial \Omega}&=(\log \sqrt{\gamma_1}-\log\sqrt{\gamma_2})|_{\partial \Omega}  \notag,
\end{align}
with $F\in H^{-1}(\Omega).$

Since $v$ is also a weak solution of the elliptic equation
$\nabla\cdot (\sqrt{\gamma_1}\sqrt{\gamma_2})\nabla
v=(\sqrt{\gamma_1}\sqrt{\gamma_2}) \cdot F$ in $\Omega$, we get the
following estimate
\begin{align}
     \|v\|_{H^{1}(\Omega)} &\lesssim   \|F\|_{H^{-1}(\Omega)}+\|v\|_{H^{1/2}(\partial\Omega)}.
\end{align}
Using interpolation theory, Theorem \ref{AL} and $\gamma_j\in
H^{\frac{3}{2}+\sigma}(\Omega)$, we get
\begin{align}
\|v\|_{H^{1/2}(\partial\Omega)}\lesssim \|v\|^{1/2}_{L^2(\partial\Omega)}\|v\|^{1/2}_{H^{1}(\partial\Omega)}
\lesssim \|\tilde\Lambda_{\gamma_1}-\tilde\Lambda_{\gamma_2}\|^{1/2}_{*}.
\end{align}
Hence, we obtain
\begin{align}\label{v}
     \|v\|_{H^{1}(\Omega)} &\lesssim   \|F\|_{H^{-1}(\Omega)}+\|\tilde\Lambda_{\gamma_1}-\tilde\Lambda_{\gamma_2}\|^{1/2}_{*}.
\end{align}

The stability will now follow after treating
$\|F\|_{H^{-1}(\Omega)}$. Following the argument in \cite{H} and
(\ref{equF}), let $g=\nabla(\log
\sqrt{\gamma_1}+\log\sqrt{\gamma_2})$ and denote by $\tilde f$ the
extension of $f\in L^2(\Omega)$ by zero to $\mathbb{R}^n$. Then for
$\varphi\in H^1_0(\Omega)$ we have
\begin{align*}
     \langle F,\varphi  \rangle &=\int_\Omega -\nabla v\nabla\overline\varphi+(g\nabla v)\overline\varphi dx\\
                &=\int_{\mathbb{R}^n} -\widetilde{\nabla v}\nabla\overline{\widetilde\varphi}+(g\widetilde{\nabla v})\overline{\widetilde\varphi} dx\\
                &=\int_{\mathbb{R}^n}\left(  (ik) \mathcal{F}\widetilde{\nabla v}+ \mathcal{F}(g\widetilde{\nabla v})   \right)\overline{\mathcal{F}\widetilde\varphi} dk.
\end{align*}
Hence
\begin{align*}
     |\langle F,\varphi\rangle|
                \leq\left( \int_{\mathbb{R}^n}\left| (ik) \mathcal{F}\widetilde{\nabla v}+ \mathcal{F}(g\widetilde{\nabla v})   \right|^2 \left(1+|k|^2\right)^{-1}dk \right)^{\frac{1}{2}} \|\tilde\varphi\|_{H^1(\mathbb{R}^n)}.
\end{align*}
Here $\mathcal{F}$ denotes the Fourier transform. Since $\gamma_i\in H^{\frac{3}{2}+\sigma}(\Omega)$, it follows that
\begin{align}\label{F}
         \|F\|^2_{H^{-1}(\Omega)}
&\leq     \int_{|k|\leq R}\left| (ik) \mathcal{F}\widetilde{\nabla v}+ \mathcal{F}(g\widetilde{\nabla v})   \right|^2 \left(1+|k|^2\right)^{-1}dk  \notag\\
          &\quad+ \int_{|k|>R}\left| (ik) \mathcal{F}\widetilde{\nabla v}+ \mathcal{F}(g\widetilde{\nabla v})   \right|^2 \left(1+|k|^2\right)^{-1}dk  \notag\\
&\lesssim   R^n \| (ik) \mathcal{F}\widetilde{\nabla v}+ \mathcal{F}(g\widetilde{\nabla v})  \|_{L^\infty(B(0,R))}^2 \notag \\
          &\quad+ \frac{1}{R^2} \|g\widetilde{\nabla v}\|_{L^2(\mathbb{R}^n)}+ \int_{|k|>R}  \left(1+|k|^2\right)^{\frac{1}{2}}\left|\mathcal{F}\widetilde{\nabla v}\right|^2 \left(1+|k|^2\right)^{-\frac{1}{2}}dk  \notag\\
&\lesssim   R^n \| (ik) \mathcal{F}\widetilde{\nabla v}+ \mathcal{F}(g\widetilde{\nabla v})  \|_{L^\infty(B(0,R))}^2 \notag \\
          &\quad+ \frac{1}{R^2} \|g\widetilde{\nabla v}\|_{L^2(\mathbb{R}^n)}+\frac{1}{R}\|\nabla v\|^2_{H^{\frac{1}{2}}(\Omega)}.
\end{align}

Now we need to estimate $\| (ik) \mathcal{F}\widetilde{\nabla v}+ \mathcal{F}(g\widetilde{\nabla v})  \|_{L^\infty(B(0,R))}^2$.
Denote $q=(ik) \widetilde{\nabla v}+(g\widetilde{\nabla v})$. Plug $u_i= \sqrt{\gamma_i}^{-1}e^{x\cdot\zeta_i}(1+\psi_i), i=1,2,$ into (\ref{int2}), we obtain that
\begin{align}
         |\mathcal{F}(q)(k)|^2
         &=\left|\int_\Omega e^{-ik\cdot x}\left( ik\nabla\left( \log\sqrt{\gamma_1}-\log\sqrt{\gamma_2})+ (  \nabla \log\sqrt{\gamma_1} \right)^2-\left(  \nabla \log\sqrt{\gamma_2} \right)^2  \right) dx\right|^2\notag\\
&\leq     \left|\int_{\Omega}   \left(\sqrt{\gamma_1}\nabla \sqrt{\gamma_2}-\sqrt{\gamma_2}\nabla \sqrt{\gamma_1} \right)\cdot \nabla \left(   \frac{1}{\sqrt{\gamma_1}\sqrt{\gamma_2}}e^{-ik\cdot x}\left( \psi_1+\psi_2+\psi_1\psi_2 \right) \right) dx\right|^2 \notag \\
&\quad+s^{-2\sigma}+s^{-1}+ \|\tilde \Lambda_{\gamma_1}-\tilde\Lambda_{\gamma_2}\|^{2\theta}_{*}
              +e^{cs} ( \|\tilde\Lambda_{\gamma_1}-\tilde\Lambda_{\gamma_2}\|_{*} +  \|\tilde\Lambda_{\gamma_1}-\tilde\Lambda_{\gamma_2}\|^2_{*}).
\end{align}

To estimate
\begin{align*}
               &\ \left|\int_{\Omega}   \left(\sqrt{\gamma_1}\nabla \sqrt{\gamma_2}-\sqrt{\gamma_2}\nabla \sqrt{\gamma_1} \right)\cdot \nabla \left(   \frac{1}{\sqrt{\gamma_1}\sqrt{\gamma_2}}e^{-ik\cdot x}\left( \psi_1+\psi_2+\psi_1\psi_2 \right) \right)  \notag dx\right|^2   \notag\\
&\lesssim      \left|\int_{\Omega}   \left(\sqrt{\gamma_1}\nabla \sqrt{\gamma_2}-\sqrt{\gamma_2}\nabla \sqrt{\gamma_1} \right)\cdot \nabla \left(   \frac{1}{\sqrt{\gamma_1}\sqrt{\gamma_2}}e^{-ik\cdot x}\right)  \left( \psi_1+\psi_2+\psi_1\psi_2 \right) dx\right|^2  \notag\\
              &\quad +\left|\int_{\Omega}   \left(\sqrt{\gamma_1}\nabla \sqrt{\gamma_2}-\sqrt{\gamma_2}\nabla \sqrt{\gamma_1} \right)\cdot \left(   \frac{1}{\sqrt{\gamma_1}\sqrt{\gamma_2}}e^{-ik\cdot x}\right)\left( \nabla\psi_1+\nabla\psi_2+\nabla\left(\psi_1\psi_2 \right)\right)   dx\right|^2  \notag\\
&=: I+II.
\end{align*}

For $I$, using Theorem \ref{w} and the definition of
$\psi_i=\sqrt{\gamma_i}\left(e^{-\frac{\phi_{is}}{2}}-\sqrt{\gamma_i}^{-1}\right)+\sqrt{\gamma_i}
e^{-\frac{\phi_{is}}{2}} w_i=:\psi^{i1}+\psi^{i2}$, we can deduce
from (\ref{psi3.6}) that
\begin{align*}
      I
      &\lesssim   \left(|k|^2+1\right)\left( \|\psi_1\|^2_{L^2(\Omega)}+  \|\psi_2\|^2_{L^2(\Omega)} + \|\psi_1\|_{L^2(\Omega)}^2\|\psi_2\|^2_{L^2(\Omega)} \right)\notag\\
         &\lesssim   \left(|k|^2+1\right)\left(s^{-2-2\sigma}+   s^{-1} \left(\|w_1\|^2_{\dot{X}^{1/2}_{\zeta_1}} +\|w_2\|^2_{\dot{X}^{1/2}_{\zeta_2}}   \right)\right)\notag\\
&\lesssim   |k|^2\left(s^{-2-2\sigma}+   s^{-1} \left(\|w_1\|^2_{\dot{X}^{1/2}_{\zeta_1}} +\|w_2\|^2_{\dot{X}^{1/2}_{\zeta_2}}    \right)\right).
\end{align*}
To estimate $II$, we divide it into two parts.
\begin{align*}
   II &\lesssim  \left|\int_{\Omega}   \left(\sqrt{\gamma_1}\nabla \sqrt{\gamma_2}-\sqrt{\gamma_2}\nabla \sqrt{\gamma_1}\right)\cdot \left(   \frac{1}{\sqrt{\gamma_1}\sqrt{\gamma_2}}e^{-ik\cdot x}\right)\left( \nabla\psi^{11}+\nabla\psi^{21}+\nabla(\psi_1\psi_2)\right)   dx\right|^2\notag\\
&\quad + \left|\int_{\Omega}   \left(\sqrt{\gamma_1}\nabla \sqrt{\gamma_2}-\sqrt{\gamma_2}\nabla \sqrt{\gamma_1} \right)\cdot \left(   \frac{1}{\sqrt{\gamma_1}\sqrt{\gamma_2}}e^{-ik\cdot x}\right)\left( \nabla\psi^{12}+\nabla\psi^{22}\right)   dx\right|^2\notag\\
&=: J_1+J_2.
\end{align*}
For $J_1$, using Lemma \ref{pre},
\begin{align*}
     J_1 &\lesssim  s^{-2\sigma}+\|\psi_1\|_{L^2}\|\nabla\psi_2\|_{L^2}+\|\psi_2\|_{L^2}\|\nabla\psi_1\|_{L^2}\notag\\
        &\lesssim   s^{-2\sigma}+\left(s^{-1-\sigma}+s^{-\frac{1}{2}}\|w_1\|_{\dot{X}^{1/2}_{\zeta_1}}  \right)\left(s^{-\sigma}+s^{\frac{1}{2}}\|w_2\|_{\dot{X}^{1/2}_{\zeta_2}}\right) \notag\\
&\quad +\left(s^{-1-\sigma}+s^{-\frac{1}{2}}\|w_2\|_{\dot{X}^{1/2}_{\zeta_2}}  \right)\left(s^{-\sigma}+s^{\frac{1}{2}}\|w_1\|_{\dot{X}^{1/2}_{\zeta_1}}\right).
\end{align*}
To estimate $J_2$, first we have
\begin{align*}
             J_2
\lesssim     \|w_1\|^2_{L^2(\Omega)}+ \|w_2\|^2_{L^2(\Omega)}+\left|\int_{\mathbb{R}^n}\Phi_B \nabla w_1  dx \right|^2 +\left|\int_{\mathbb{R}^n} \Phi_B \nabla w_2  dx \right|^2.
\end{align*}
Note that since $\gamma_j\in H^{3/2}(\Omega)$, the function $\Phi_B$
has compact support and is in the space $H^{1/2}(\mathbb{R}^n)$.
Then $\left|\int_{\mathbb{R}^n} \Phi_B \nabla w_1  dx \right|^2
\lesssim \|\Phi_B\|^2_{H^{1/2}(\mathbb{R}^n)}
\|\Phi_Bw_1\|^2_{H^{1/2}(\mathbb{R}^n)}$. We derive
\begin{align*}
J_2 \lesssim     s^{-1}\left(\|w_1\|^2_{\dot{X}^{1/2}_{\zeta_1}} +\|w_2\|^2_{\dot{X}^{1/2}_{\zeta_2}} \right)  +\left(\|w_1\|^2_{\dot{X}^{1/2}_{\zeta_1}} +\|w_2\|^2_{\dot{X}^{1/2}_{\zeta_2}} \right)
\end{align*}
by applying $ \|w\|_{L^2(\Omega)}\lesssim
s^{-1/2}\|w\|_{\dot{X}^{1/2}_{\zeta_1}}$ from Lemma \ref{sobolev}.

Based on the argument above, we have the estimate
\begin{align}\label{fourier}
        |\mathcal{F}(q)(k)|^2& \lesssim  |k|^2\left(s^{-2-2\sigma}+   s^{-1} \left(\|w_1\|^2_{\dot{X}^{1/2}_{\zeta_1}} +\|w_2\|^2_{\dot{X}^{1/2}_{\zeta_2}}    \right)\right) + \left(\|w_1\|^2_{\dot{X}^{1/2}_{\zeta_1}} +\|w_2\|^2_{\dot{X}^{1/2}_{\zeta_2}} \right) \notag\\
&\quad+s^{-\frac{1}{2}-\sigma}\|w_j\|_{\dot{X}^{1/2}_{\zeta_j}}+\|w_1\|_{\dot{X}^{1/2}_{\zeta_1}}\|w_2\|_{\dot{X}^{1/2}_{\zeta_2}} \notag\\
&\quad+s^{-2\sigma}+s^{-1}+ \|\tilde
\Lambda_{\gamma_1}-\tilde\Lambda_{\gamma_2}\|^{2\theta}_{*}
             +e^{cs} ( \|\tilde\Lambda_{\gamma_1}-\tilde\Lambda_{\gamma_2}\|_{*} +  \|\tilde\Lambda_{\gamma_1}-\tilde\Lambda_{\gamma_2}\|^2_{*}).
\end{align}
Integrating on both sides of (\ref{fourier}), we get
\begin{align*}
        |\mathcal{F}(q)(k)|^2
&\lesssim              |k|^2\left(\lambda^{-2-2\sigma}+   \frac{1}{\lambda}\int_{S^{n-1}}\int^{2\lambda}_{\lambda}s^{-1} \left(\|w_1\|^2_{\dot{X}^{1/2}_{\zeta_1}} +\|w_2\|^2_{\dot{X}^{1/2}_{\zeta_2}}    \right)\right) dsd\eta    \notag\\ &\quad +\frac{1}{\lambda}\int_{S^{n-1}}\int^{2\lambda}_{\lambda}\left(\|w_1\|^2_{\dot{X}^{1/2}_{\zeta_1}} +\|w_2\|^2_{\dot{X}^{1/2}_{\zeta_2}} \right)dsd\eta \notag\\
&\quad + \frac{1}{\lambda}\int_{S^{n-1}} \int^{2\lambda}_{\lambda}\left(s^{-\frac{1}{2}-\sigma}\|w_j\|_{\dot{X}^{1/2}_{\zeta_j}}+\|w_1\|_{\dot{X}^{1/2}_{\zeta_1}}\|w_2\|_{\dot{X}^{1/2}_{\zeta_2}}\right)    dsd\eta\notag\\
&\quad +\lambda^{-2\sigma}+\lambda^{-1}+ \|\tilde
\Lambda_{\gamma_1}-\tilde\Lambda_{\gamma_2}\|^{2\theta}_{*}
              +e^{c\lambda}   ( \|\tilde\Lambda_{\gamma_1}-\tilde\Lambda_{\gamma_2}\|_{*} +  \|\tilde\Lambda_{\gamma_1}-\tilde\Lambda_{\gamma_2}\|^2_{*}).
\end{align*}
Applying estimate (\ref{normw})
\begin{align*}
 \frac{1}{\lambda}\int_{S^{n-1}}\int^{2\lambda}_{\lambda}  \|w\|^2_{\dot{X}^{1/2}_{\zeta}}  ds d\eta\lesssim \lambda^{-2\sigma}+\lambda^{-1},
\end{align*}
we have
\begin{align}\label{fourier1}
 |\mathcal{F}(q)(k)|^2
&\lesssim        |k|^2\left(\lambda^{-1-2\sigma}+\lambda^{-2}
\right)
+\lambda^{-2\sigma}+\lambda^{-1}\notag\\
&\quad + \|\tilde \Lambda_{\gamma_1}-\tilde\Lambda_{\gamma_2}\|^{2\theta}_{*}
              +e^{c\lambda}  ( \|\tilde\Lambda_{\gamma_1}-\tilde\Lambda_{\gamma_2}\|_{*} +  \|\tilde\Lambda_{\gamma_1}-\tilde\Lambda_{\gamma_2}\|^2_{*}).
\end{align}
Varying $\eta$ in a small conic neighborhood $U_\eta\in S^{n-1}$, we
get the estimate (\ref{fourier1}) uniformly for all $k\in E=\{k\in
\mathbb{R}^n: k\ \hbox{orthogonal to some}\ \tilde\eta\in U_\eta
\}$.

Fixed $R>0$ and $k\in \mathbb{R}^n$. Let $f(k)=\mathcal{F}( q)(Rk)$. Since $q$ is compactly supported, $\mathcal{F}( q)$ is analytic by the Paley-Wiener theorem and
\begin{align*}
     |D^\alpha f(k)|\leq \|q\|_{L^1(\Omega)}\frac{R^{|\alpha|}}{(\text{diam}(\Omega)^{-1})^{|\alpha|}}\leq C \frac{R^{|\alpha|}}{\alpha!(\text{diam}(\Omega)^{-1})^{|\alpha|}}\alpha! \leq C\frac{e^{nR}}{(\text{diam}(\Omega)^{-1})^{|\alpha|}}\alpha!
\end{align*}
for any $\alpha\in \mathbb{N}^n$. Let $D=B(0,2)$ and $\tilde E=E\cap B(0,1)$ with $M=Ce^{nR}$ and $\rho=\text{diam}(\Omega)^{-1}$.
From Proposition \ref{Ve}, there exists $\tilde \theta\in (0,1)$ such that
\begin{align}\label{qk}
     |\mathcal{F}(q)(k)|=|f(k/R)|\leq Ce^{nR(1-\tilde\theta)}\|f\|^{\tilde\theta}_{L^\infty(\tilde E)}\leq Ce^{nR(1-\tilde\theta)}\|\mathcal{F}(q)(k)\|^{\tilde\theta}_{L^\infty(E)}
\end{align}
for all $k\in B(0,R)$.

Using (\ref{qk}), together with (\ref{fourier1}) and (\ref{F}), we get
\begin{align*}
          \|F\|^2_{H^{-1}(\Omega)}
&\lesssim       R^n e^{2nR(1-\tilde\theta)}\Big(   \lambda^{-2\sigma}+\lambda^{-1}
             +\|\tilde \Lambda_{\gamma_1}-\tilde\Lambda_{\gamma_2}\|^{2\theta}_{*}
             \\
             &\quad+e^{c\lambda}  ( \|\tilde\Lambda_{\gamma_1}-\tilde\Lambda_{\gamma_2}\|_{*}
              + \|\tilde\Lambda_{\gamma_1}-\tilde\Lambda_{\gamma_2}\|^2_{*})  \Big)^{\tilde\theta}+R^{-1}
\end{align*}
if $\lambda >R^2>1$.
Thus,
\begin{align}\label{Fest}
          \|F\|^{\frac{2}{\tilde\theta}}_{H^{-1}(\Omega)}
&\lesssim    R^{\frac{n}{\tilde\theta}} e^{2nR\frac{1-\tilde\theta}{\tilde\theta}} \lambda^{-2\beta}+R^{\frac{n}{\tilde\theta}} e^{2nR\frac{1-\tilde\theta}{\tilde\theta}}\|\tilde \Lambda_{\gamma_1}-\tilde\Lambda_{\gamma_2}\|^{2\theta}_{*}  \notag\\
     &\quad +R^{\frac{n}{\tilde\theta}} e^{2nR\frac{1-\tilde\theta}{\tilde\theta}}e^{c\lambda}  ( \|\tilde\Lambda_{\gamma_1}-\tilde\Lambda_{\gamma_2}\|_{*} +  \|\tilde\Lambda_{\gamma_1}-\tilde\Lambda_{\gamma_2}\|^2_{*})+R^{-\frac{1}{\tilde\theta}}.
\end{align}
Here we denote
\begin{align*}
    \left\{
       \begin{array}{ll}
     \beta=\sigma\ \ \ &\hbox{if $0<\sigma\leq \frac{1}{2}$},\\
     \beta=\frac{1}{2}\ \ \ \ \ &\hbox{if $\frac{1}{2}<\sigma<1$  }.
    \end{array}
     \right.
\end{align*}
Choosing
\begin{align*}
      &\lambda= \left(R^{n+1} e^{2nR(1-\tilde\theta)}\right)^{\frac{1}{2\beta\tilde\theta}}
\end{align*}
such that
\begin{align*}
      &R^{\frac{n}{\tilde\theta}}e^{2nR\frac{1-\tilde\theta}{\tilde\theta}}\lambda^{-2\beta}=R^{-\frac{1}{\tilde\theta}}, \end{align*}
the estimate (\ref{Fest}) is bounded by
\begin{align}\label{FF}
          \|F\|^{\frac{2}{\tilde\theta}}_{H^{-1}(\Omega)}
&\lesssim    R^{\frac{n}{\tilde\theta}} e^{2nR\frac{1-\tilde\theta}{\tilde\theta}}\|\tilde \Lambda_{\gamma_1}-\tilde\Lambda_{\gamma_2}\|^{2\theta}_{*}\notag \\
&\quad+ R^{\frac{n}{\tilde\theta}}  e^{2nR\frac{1-\tilde\theta}{\tilde\theta}}e^{c\lambda}  ( \|\tilde\Lambda_{\gamma_1}-\tilde\Lambda_{\gamma_2}\|_{*}
  +  \|\tilde\Lambda_{\gamma_1}-\tilde\Lambda_{\gamma_2}\|^2_{*})+R^{-\frac{1}{\tilde\theta}}.
\end{align}
Using the fact that
\begin{align*}
R^{\frac{n}{\tilde\theta}} e^{2nR\frac{1-\tilde\theta}{\tilde\theta}+c\lambda}
            &=R^{\frac{n}{\tilde\theta}} e^{2nR\frac{1-\tilde\theta}{\tilde\theta}+c   \left(R^{n+1} e^{2nR(1-\tilde\theta)}\right)^{\frac{1}{2\beta\tilde\theta}}  } \\
&\leq \exp\left(e^{\left[\frac{n}{\tilde\theta}+2n\frac{1-\tilde\theta}{\tilde\theta}+c+\frac{n+1}{2\beta\tilde\theta}+\frac{n(1-\tilde\theta)}{\beta\tilde\theta}\right]R}\right)\ \ \ \hbox{for all $R>0$}.
\end{align*}
Setting
$K=\frac{n}{\tilde\theta}+2n\frac{1-\tilde\theta}{\tilde\theta}+c+\frac{n+1}{2\beta\tilde\theta}+\frac{n(1-\tilde\theta)}{\beta\tilde\theta}$,
(\ref{FF}) leads to
\begin{align}
          \|F\|^{\frac{2}{\tilde\theta}}_{H^{-1}(\Omega)}
\lesssim    e^{e^{KR}}\left( \|\tilde \Lambda_{\gamma_1}-\tilde\Lambda_{\gamma_2}\|^{2\theta}_{*}+ \|\tilde\Lambda_{\gamma_1}-\tilde\Lambda_{\gamma_2}\|_{*}+\|\tilde\Lambda_{\gamma_1}-\tilde\Lambda_{\gamma_2}\|^2_{*}\right)  +R^{-\frac{1}{\tilde\theta}}.
\end{align}

The arguments above are valid if $\lambda\geq \lambda_0$. There
exists a small $\delta$ such that if
$\|\tilde\Lambda_{\gamma_1}-\tilde\Lambda_{\gamma_2}\|_{*}<\delta$
and $R=\frac{1}{K}\log{|\log
\|\tilde\Lambda_{\gamma_1}-\tilde\Lambda_{\gamma_2}\|_{*}^\theta|}$,
we have $\lambda\geq \lambda_0$. To be more precise, if
$$
\lambda_0\leq \lambda=\left(R^{n+1} e^{2nR(1-\tilde\theta)}\right)^{\frac{1}{2\beta\tilde\theta}}\leq \left(e^{R(n+1)} e^{2nR(1-\tilde\theta)}\right)^{\frac{1}{2\beta\tilde\theta}},
$$
then
$$
   R\geq \frac{2\beta\tilde\theta}{3n+1-2n\tilde \theta}\log{\lambda_0}=:R_0.
$$
We take $0<\delta\leq \delta_0<1$ with
$
    \delta_0^\theta\leq e^{-e^{K\exp{R_0}}}.
$ Thus
\begin{align}\label{4.8}
          \|F\|_{H^{-1}(\Omega)}
\lesssim    \left( \|\tilde \Lambda_{\gamma_1}-\tilde\Lambda_{\gamma_2}\|^{\theta}_{*}+ \|\tilde \Lambda_{\gamma_1}-\tilde\Lambda_{\gamma_2}\|^{1-\theta}_{*} +\frac{1}{K}\log{|\log \|\tilde\Lambda_{\gamma_1}-\tilde\Lambda_{\gamma_2}\|_{*}^\theta|}^{-\frac{1}{\tilde\theta}} \right)^{\frac{\tilde\theta}{2}}.
\end{align}
For any $f\in L^\infty(\mathbb{R}^n)$ and $0<\tilde\sigma<1$, we deduce that
$$
    |f(x)|^{\frac{n}{1-\tilde\sigma}}\leq \|f\|_{L^\infty(\mathbb{R}^n)}^{\frac{n}{1-\tilde\sigma}-2} |f(x)|^2
$$
for almost every $x\in \mathbb{R}^n$. Then we have
\begin{align}\label{C1}
    \|\gamma_1-\gamma_2\|_{W^{1,\frac{n}{1-\tilde\sigma}}(\Omega)}\lesssim \|\gamma_1-\gamma_2\|_{H^1(\Omega)}^{\frac{2(1-\tilde\sigma)}{n}}.
\end{align}
From Theorem 5 in Ch. 5 in \cite{evan}, we obtain that
\begin{align}\label{C0}
\|\gamma_1-\gamma_2\|_{C^{0,\tilde\sigma}(\overline\Omega)}\lesssim \|\gamma_1-\gamma_2\|_{W^{1,\frac{n}{1-\tilde\sigma}}(\Omega)}.
\end{align}
Applying (\ref{v}), (\ref{4.8}), (\ref{C1}) and (\ref{C0}), the
estimate
\begin{align*}
   \|\gamma_1-\gamma_2\|_{C^{0,\tilde\sigma}(\overline\Omega)}
&\lesssim  \left( \|\tilde \Lambda_{\gamma_1}-\tilde\Lambda_{\gamma_2}\|^{\theta}_{*}+ \|\tilde \Lambda_{\gamma_1}-\tilde\Lambda_{\gamma_2}\|^{1-\theta}_{*} +\frac{1}{K}\log{|\log \|\tilde\Lambda_{\gamma_1}-\tilde\Lambda_{\gamma_2}\|_{*}^\theta|}^{-\frac{1}{\tilde\theta}} \right)^{\frac{\tilde\theta(1-\tilde\sigma)}{n}}
\end{align*}
holds.

Now if
$\|\tilde\Lambda_{\gamma_1}-\tilde\Lambda_{\gamma_2}\|_{*}\geq\delta>0$,
then we have
\begin{align}\label{4.9}
          \|\gamma_1-\gamma_2\|_{C^{0,\tilde\sigma}(\overline\Omega)}
\leq  \frac{C}{\delta^{\frac{\theta\tilde\theta(1-\tilde\sigma)}{n}}}   \delta^{\frac{\theta\tilde\theta(1-\tilde\sigma)}{n}}
\lesssim \|\tilde \Lambda_{\gamma_1}-\tilde\Lambda_{\gamma_2}\|^{\frac{\theta\tilde\theta(1-\tilde\sigma)}{n}}_{*}
\end{align}
for some $C>0$. The proof of Theorem 1.1 is completed.

\bigskip

\textbf{Acknowlegments.}
The author would like to thank professor Gunther Uhlmann for his encouragements and helpful discussions. The author also would like to thank the anonymous referee for his or her comments which have contributed to improve this manuscript. The author is partially supported by NSF.




\end{document}